\documentclass[12pt]{article}

\usepackage{natbib}
\usepackage{amsmath}
\newtheorem{theorem}{Theorem}[section]
\newtheorem{lemma}[theorem]{Lemma}
\newtheorem{proof}[theorem]{Proof}

\begin{document}

\title{A Blockwise Descent Algorithm for Group-penalized Multiresponse and Multinomial Regression}

\author{Noah Simon \and Jerome Friedman \and Trevor Hastie}

\maketitle

\begin{abstract}
In this paper we purpose a blockwise descent algorithm for group-penalized multiresponse regression.  Using a quasi-newton framework we extend this to group-penalized multinomial regression. We give a publicly available implementation for these in \texttt{R}, and compare the speed of this algorithm to a competing algorithm ---  we show that our implementation is an order of magnitude faster than its competitor, and can solve gene-expression-sized problems in real time.
\end{abstract}

\section{Introduction}\label{sec:intro}
Consider the usual linear regressions framework with $y$ an $n$-vector of responses and $X$, an $n$ by $p$ matrix of covariates. Traditionally, problems involve $n<p$ and it is standard to estimate a regression line with least squares. In many recent applications (genomics and advertising, among others) we have $p\gg n$, and standard regression fails. In these cases, one is often interested in a solution involving only few covariates. Toward this end, \citet{tibs1996} proposed the Lasso: to find our regression coefficients by solving the regularized least squares problem
\[
\hat{\beta} = \operatorname{argmin}_{\beta} \frac{1}{2}\left\|y - X\beta\right\|_2^2 + \lambda\left\|\beta\right\|_1
\]
This method regularizes $\beta$ by trading off ``goodness of fit'' for a reduction in ``wildness of coefficients'' --- it also has the effect of giving a solution with few nonzero entries in $\beta$. This was generalized by \citet{YL2007} to deal with grouped covariates; they propose to solve
\begin{equation}\label{eq:glasso}
\hat{\beta} = \operatorname{argmin}_{\beta} \frac{1}{2}\left\|y - X\beta\right\|_2^2 + \lambda\sum_k\left\|\beta_{I(k)}\right\|_2
\end{equation}
where the covariates are partitioned into disjoint groups and $I(k)$ denotes the indices of the $k$th group of covariates (i.e., $\beta_{I(k)}$ indicates the sub-vectors of $\beta$ corresponding to group $k$). This approach gives a solution with few non-zero groups.\\

Now, instead of usual linear regression, one might be interested in multiresponse regression --- instead of an $n$-vector, $Y$ is an $n\times M$ matrix, and $\beta$ is a $p\times M$ matrix. In some cases one might believe that our response variables are related, in particular that they have roughly the same set of important explanatory variables, a subset of all predictor variables measured (i.e., in each row of $\beta$ either all of the elements are zero or all are
non-zero.). A number of authors have similar suggestions for this problem (\citealp{obozinski2007}, \citealp{argyriou2007}, among others) which build on the group-lasso idea of \citet{YL2007}; to use a group-penalty on rows of $\beta$:
\begin{equation}\label{eqn:mrlasso}
\hat{\beta} = \operatorname{argmin}_{\beta} \frac{1}{2}\left\|Y - X\beta\right\|_F^2 + \lambda\sum_{k=1}^p\left\|\beta_{k \cdot}\right\|_2
\end{equation}
where $\beta_{k \cdot}$ refers to the $k$th row of $\beta$ (likewise, we will, in the future, use $\beta_{\cdot m}$ to denote the $m$th column of $\beta$). For the remainder of this paper we will refer to Equation~\ref{eqn:mrlasso} as the ``multiresponse lasso''.

Multinomial regression (via a generalized linear model) is also initimately related to multiresponse regression. In particular, common methods for finding the MLE in non-penalized multinomial regression (eg Newton-Raphson) reduce the problem to that of solving a series of weighted multiresponse regression problems. In multinomial regression, $\beta$ is again an $n\times M$ matrix, where $M$ is now the number of classes --- the $(k,m)$ entry gives the contribution of variable $x_k$ to class $m$. More specifically, the probability of an observation with covariate vector $x$ belonging to a class $l$ is parametrized as
\[
\operatorname{P}\left(y = l \middle| x\right) = \frac{\exp\left(x^{\top}\beta_{\cdot l}\right)}{\sum_{m\leq M}\exp\left(x^{\top}\beta_{\cdot m}\right)}.
\]
 As in linear regression, one might want to fit a penalized version of this model. Standard practice has been to fit
\[
\hat{\beta} = \operatorname{argmin}_{\beta}-\ell\left(y,X\beta\right) + \lambda\left\|\beta\right\|_1.
\]
where $\ell$ is the multinomial log-likelihood. While this does give sparsity in $\beta$ it may give a very different set of non-zero coefficients in each class. We instead discuss using
\begin{equation}\label{eq:groupedMultinom}
\hat{\beta} = \operatorname{argmin}_{\beta}-\ell\left(y,X\beta\right) + \lambda\sum_{k=1}^p\left\|\beta_{k\cdot}\right\|_2.
\end{equation}
 which we will refer to as the group-penalized multinomial lasso. Because each term in the penalty sum is non-differentiable only when all elements of the vector $\beta_{k\cdot}$ are $0$, this formulation has the advantage of giving the same nonzero coefficients in each class. It was recently proposed by \citet{vincent2012}. When the true underlying model has the same (or similar) non-zero coefficient structure across classes, this model can improve prediction accuracy. Furthermore, this approach can lead to more interpretable model estimates.\\

One downside of this approach is that minimization of the criterion in Equation~\ref{eq:groupedMultinom} requires new tools. The criterion is convex so for small to moderate sized problems interior point methods can be used --- however, in many applications there will be many features ($p > 10,000$) and possibly a large number of classes and/or observations. For the usual lasso, there are coordinate descent and first order algorithms that can scale to much larger problem sizes. These algorithms must be adjusted for the grouped multinomial lasso problem.

In this paper we discuss an efficient block coordinate descent algorithm to fit the group-penalized multiresponse lasso, and group-penalized multinomial lasso. The algorithm in this paper is the multiresponse analog to \citet{FHT2010}. In particular, we have incorporated this algorithm into \texttt{glmnet} \citep{glmnet} a widely used \texttt{R} package for solving penalized regression problems. 

\section{Penalized multiresponse regression}\label{sec:MR}
We first consider our Gaussian objective, Equation~\ref{eqn:mrlasso}:
\[
\min \frac{1}{2}\left\|Y - X\beta\right\|_F^2 + \lambda\sum_{k\leq p}\left\|\beta_{k \cdot}\right\|_2
\]
This can be minimized by blockwise coordinate descent (one row of $\beta$ at a time). Consider a single $\beta_{k\cdot}$ with fixed $\beta_{j\cdot}$ for all $j\neq k$. Our objective becomes
\[
\min \frac{1}{2}\left\|R_{-k} - X_{\cdot k}\beta_{k \cdot}\right\|_F^2 + \lambda\left\|\beta_{k \cdot}\right\|_2
\]
where $X_{\cdot k}$ refers to the $k$th column of $X$, and $R_{-k} = Y - \sum_{j\neq k}X_{\cdot j}\beta_{j \cdot}$ is the partial residual. If we
take a subderivative with respect $\beta_{k \cdot}$, then we get that
$\hat{\beta}_{k \cdot}$ satisfies
\begin{equation}\label{eq:subD}
\left\|X_{\cdot k}\right\|_2^2\hat{\beta}_{k \cdot}^{\top} -
X_{\cdot k}^{\top}R_{-k} + \lambda S(\hat{\beta}_{k \cdot}) = 0
\end{equation}
where $S(\hat{\beta}_{k \cdot})$ is its sub-differential
\[
S(a) 
\begin{cases}
 = \frac{a}{||a||_2}, & \textrm{if } a\neq 0 \\
\in \left\{u \textrm{ s.t. } ||u||_2\leq 1\right\}, & \textrm{if } a = 0
\end{cases}
\]
From here, simple algebra gives us that
\begin{equation}\label{eq:MRfit}
\hat{\beta}_{k \cdot} = \frac{1}{\left\|X_{\cdot k}\right\|_2^2}\left(1 - \frac{\lambda}{\left\|X_{\cdot k}^{\top}R_{-k}\right\|_2}\right)_{+}X_{\cdot k}^{\top}R_{-k}
\end{equation}
where $(a)_+ = \max(0,a)$. By cyclically applying these updates we can minimize our objective. The convergence guarantees of this style of algorithm are discussed in \citet{tseng2001}.\\
\\
{\bf Gaussian Algorithm}\\
Now we have the following very simple algorithm.  We can include all
the other bells and whistles from \texttt{glmnet} as well
\begin{enumerate}
\item Initialize $\beta = \beta_0$, $R = Y - X\beta_0$.
\item Iterate until convergence: for $k=1,\ldots,p$
\begin{enumerate}
\item Update $R_{-k}$ by
\[ 
R_{-k} = R + X_{\cdot k}\beta_{k \cdot}.
\]
\item Update $\beta_{k \cdot}$ by
\[
\beta_{k \cdot} \leftarrow
\frac{1}{\left\|X_{\cdot k}\right\|_2^2}\left(1 -
  \frac{\lambda}{\left\|X_{\cdot k}^{\top}R_{-k}\right\|_2}\right)_{+}X_{\cdot k}^{\top}R_{-k}.
\]
\item Update $R$ by
\[
R = R_{-k} - X_{\cdot k}\beta_{k \cdot}.
\]
\end{enumerate}
\end{enumerate}
Note that in step $(2b)$ if $\left\|X_{\cdot k}^{\top}R_{-k}\right\|_2 \leq \lambda$, the update is simply $\beta_{k\cdot} \leftarrow {\bf 0}$.\\

If we would like to include intercept terms in the regression, we need only mean center the columns of $X$ and $Y$ before carrying out the algorithm (this is equivalent to a partial minimization with respect to the intercept term).

\section{Extension to multinomial regression}
We now extend this idea to the multinomial setting. Suppose we have $M$
different classes.  In this setting,
$Y$ is an $n\times M$ matrix of zeros and ones --- the $i$th row has a single
one corresponding to the class of observation $i$. In a multinomial
generalized linear model one assumes that the probability of observation $i$ coming from class $m$ has the form
\[
p_{i,m} = \frac{\operatorname{exp}\left(\eta_{i,m}\right)}{\sum_{l\leq M}\operatorname{exp}\left(\eta_{i,l}\right)}
\]
with
\[
\eta_{i,m} = X_{i\cdot}\beta_{\cdot m}.
\]
This is the symmetric parametrization. From here we get the multinomial log-likelihood
\[
\ell\left({\bf p}\right) = \sum_{i=1}^n \sum_{m=1}^M y_{i,m}\log \left(p_{i,m}\right)
\]
which we can rewrite as
\begin{align*}
\ell\left({\boldsymbol \eta}\right) &= \sum_{i=1}^n \sum_{m=1}^M y_{i,m}\left[\eta_{i,m}
  - \log\left(\sum_{l\leq M}\operatorname{exp}\left(\eta_{i,l}\right)\right)\right]\\
&= \sum_{i=1}^n \left[\sum_{m=1}^M y_{i,m}\eta_{i,m}
  - \log\left(\sum_{l\leq M}\operatorname{exp}\left(\eta_{i,l}\right)\right)\right]
\end{align*}
since $\sum_{m=1}^M y_{i,m} = 1$ for each $i$. Thus our final minimization problem is
\begin{equation}\label{eq:penalizedMultinomial}
\min -\sum_{i=1}^n \left[\sum_{m=1}^M y_{i,m}X_{i\cdot}\beta_{\cdot m}
  - \log\left(\sum_{l\leq M}\operatorname{exp}\left(X_{i\cdot}\beta_{\cdot l}\right)\right)\right] + \lambda\sum_{k\leq p}\left\|\beta_{k\cdot}\right\|_2
\end{equation}
\subsection{Uniqueness}\label{sec:unique}
Before proceeding, we should note that, because we choose to use the symmetric version of the multinomial log-likelihood, without our penalty (or with $\lambda=0$) the solution to this objective is never unique. If we add or subtract a constant to an entire row of $\beta$, the unpenalized objective is unchanged. To see this, consider replacing our estimate for the $k$th row $\hat{\beta}_{k\cdot}$ by $\hat{\beta}_{k\cdot} + \delta{\bf 1}^{\top}$ (for some scalar $\delta$). Then we have
\begin{align*}
\hat{P}_{\delta}\left(y = l \middle| x\right) &= \frac{\operatorname{exp}\left(x^{\top}\hat{\beta}_{\cdot l} + x_k\delta\right)}{\sum_{m\leq M}\operatorname{exp}\left(x^{\top}\hat{\beta}_{\cdot m} + x_k\delta\right)}\\
&=\frac{\operatorname{exp}\left(x_k\delta\right)\operatorname{exp}\left(x^{\top}\hat{\beta}_{\cdot l}\right)}{\operatorname{exp}\left( x_k\delta\right)\sum_{m\leq M}\operatorname{exp}\left(x^{\top}\hat{\beta}_{\cdot m}\right)}\\
&= \hat{P_{0}}\left(y = l \middle| x\right)
\end{align*}
Now, as the unpenalized loss is entirely determined by the estimated probabilities and the outcomes, this result tells us that the row means do not affect the unpenalized loss. This is not the case for our penalized problem --- here, the row means are all $0$.

\begin{lemma}
For a given $X$ matrix, $y$ vector, and $\lambda > 0$, let $\beta^{*}$ denote the solution to the minimization of (\ref{eq:penalizedMultinomial}). Let $\mu^*$ be the vector of row-means of $\beta^*$.\\
\\
We have that $\mu^* = {\bf 0}.$
\end{lemma}
\begin{proof}
Let $\beta^{**} = \beta^{*} - {\bf 1} {\mu^{*}}^{\top}$ be the ``row-mean centered'' version of $\beta^{*}$. Plugging these in to the penalized log-likelihood in Equation~\ref{eq:penalizedMultinomial}, we see that the difference between two penalized log-likelihoods is
\begin{align*}
0 &\leq L\left(\beta^{*}\right) - L\left(\beta^{**}\right)\\
&= \lambda\sum_{k\leq p}\left[\left\|\beta^*_{k\cdot}\right\|_2 - \left\|\beta^{**}_{k\cdot}\right\|_2\right]\\
&= \lambda\sum_{k\leq p}\left[\left\|\beta^*_{k\cdot}\right\|_2 - \sqrt{\left\|\beta^*_{k\cdot}\right\|_2^2 + \left\|\mu^*_{k\cdot}\right\|_2^2} \right]\\
&\leq 0
\end{align*}
Thus $\left\|\mu^*_{k\cdot}\right\|_2 = 0$, so $\mu^{*} = {\bf 0}$.
\end{proof}


\subsection{Multinomial optimization}
This optimization problem is nastier than its Gaussian
counterpart, as the coordinate updates no longer have a closed form solution. However, as in \texttt{glmnet} we can use an approximate Newton scheme and optimize our multinomial objective by repeatedly approximating with a
quadratic and minimizing the corresponding gaussian problem.\\

We begin with $\tilde{\beta}$, some inital guess of $\beta$. From this estimate, we can find an estimate of our probabilities:
\[
p_{i,m} = \frac{\operatorname{exp}\left(X_{i\cdot}\tilde\beta_{\cdot m}\right)}{\sum_{l\leq M}\operatorname{exp}\left(X_{i\cdot}\tilde\beta_{\cdot l}\right)}
\]

Now, as in standard Newton-Raphson, we calculate the first and second derivatives of our log-likelihood
(in $\eta$). We see that
\[
\frac{\partial \ell}{\partial{\eta_{i,m}}} = y_{i,m} - p_{i,m}
\]
For the second derivatives, as usual we get ``independence'' between observations,
i.e., if $j\neq i$, for any $m$, $l$
\[
\frac{\partial \ell}{\partial\eta_{i,m}\partial\eta_{j,l}} = 0.
\]
We also have our usual Hessian within observation
\begin{equation}\label{eq:hess1}
\frac{\partial \ell}{\partial\eta_{i,m}\partial\eta_{i,l}} = p_{i,m}p_{i,l}
\end{equation}
for $m\neq l$, and
\begin{equation}\label{eq:hess2}
\frac{\partial \ell}{\partial\eta_{i,m}^2} = -p_{i,m}(1-p_{i,m}).
\end{equation}
Let $H_i$ denote the within observation Hessian. By combining (\ref{eq:hess1}) and (\ref{eq:hess2}) we see that
\[
-H_i = \operatorname{diag}\left(p_{i\cdot}\right) - p_{i\cdot}^{\top}p_{i\cdot}
\]
and we can write out a second order Taylor series approximation to our log-likelihood (centered around some value $\tilde{\beta}$) by
\begin{align*}
-\ell(\beta) &\approx -\ell(\tilde{\beta}) - \operatorname{trace}\left[(\beta - \tilde{\beta})^{\top}
X^{\top}\left(Y - P\right)\right]\\
&- \frac{1}{2}\sum_{i=1}^n X_{i\cdot}\left(\beta -
  \tilde{\beta}\right) H_i \left(\beta -
  \tilde{\beta}\right)^{\top}X_{i\cdot}^{\top}
\end{align*}
Because we have independence of observations, the second order term has decoupled into the sum of simple quadratic forms. Unfortunately, because each of these $H_i$ are different (and not even
full rank), using this quadratic approximation instead of the original
log-likelihood would still be difficult. We would like to find a
simple matrix which dominates all of the $-H_i$. To this end, we show that $-H_i \preceq tI$ where 
\[
t = 2\max_{i,j}\left\{p_{i,j}\left(1-p_{i,j}\right)\right\}\leq 1/2
\]
\begin{lemma}
Let ${\bf p}$ be an $M$-vector of probabilities and define
\[
-H = \operatorname{diag}\left({\bf p}\right) - {\bf p}{\bf p}^{\top}.
\]
We have that $-H\preceq \operatorname{max}_{l}\left[2{\bf p}_l(1-{\bf p}_l)\right]I$
\end{lemma}

\begin{proof}
Define $D$ by
\[
D = \operatorname{diag}\left[2{\bf p}(1-{\bf p})\right] + H
\]
We would like to show that $D$ diagonally dominant and thus positive semi-definite. Toward this end, choose some $m \leq M$.
\begin{align*}
\sum_{l\neq m} |D_{m,l}| &= \sum_{l\neq m} {\bf p}_l{\bf p}_m\\
&= {\bf p}_m \sum_{l\neq m} {\bf p}_l\\
&= {\bf p}_m \left(1-{\bf p}_m\right)\\
&= D_{m,m}
\end{align*}
Thus $D$ is positive semi-definite, and so $-H\preceq \operatorname{diag}\left[2{\bf p}(1-{\bf p})\right]$. Furthermore, $\operatorname{diag}\left[2{\bf p}(1-{\bf p})\right] \preceq  \operatorname{max}_{l}\left[2{\bf p}_l(1-{\bf p}_l)\right]I$ so we also have
\[
-H\preceq \operatorname{max}_{l}\left[2{\bf p}_l(1-{\bf p}_l)\right]I
\]

\end{proof}
Now if we consider $t_i = 2\max_{j}\left\{p_{i,j}\left(1-p_{i,j}\right)\right\}$, then the preceeding lemma gives us that $-H_i \preceq t_i I$, and since for all $i$ $t_i \leq t$, we have our majorization $-H_i \preceq tI$.\\

If we replace our original Hessian with this majorizing approximation, with some algebraic manipulation we can reduce our problem to the gaussian framework. Furthermore, because our new approximation dominates the Hessian, we still enjoy nice convergence properties. Toward this end, we write a new majorizing quadratic approximation (by replacing each $-H_i$ by $tI$)
\begin{align*}
-\ell(\beta) &\leq -\ell(\tilde{\beta}) -  \operatorname{trace}\left[(\beta - \tilde{\beta})^{\top}
X^{\top}\left(Y - P\right)\right]\\
&+\frac{t}{2}\operatorname{trace}\left[\sum_{i=1}^n \left(\beta -
  \tilde{\beta}\right)^{\top}X_{i\cdot}^{\top} X_{i\cdot}\left(\beta -
  \tilde{\beta}\right)\right]
\end{align*}
which, with simple algebra, becomes
\[
-\ell(\beta) \leq -\ell(\tilde{\beta}) -  \operatorname{trace}\left[(\beta - \tilde{\beta})^{\top}
X^{\top}\left(Y - P\right)\right] + \frac{t}{2}\left\|X\left(\beta -
  \tilde{\beta}\right)\right\|_F^2
\]
Completing the square, we see that minimizing this with the addition of our penalty is equivalent to
\[
\min_{\beta} \frac{1}{2}\left\|X\left(\beta -
  \tilde{\beta}\right) - (Y - P)/t\right\|_F^2 + \lambda/t\sum\left\|\beta_{i\cdot}\right\|_2
\]
Notice that we have reduced the multinomial problem to our gaussian framework. Our plan of attack is to repeatedly approximate our loss by this penalized quadratic (centered at our current estimate of $\beta$), and minimize this loss to update our estimate of $\beta$.\\
\\
{\bf Multinomial Algorithm}\\
Combining the outer loop steps with our gaussian algorithm we have the
following simple algorithm.  We can still include all
the other bells and whistles from \texttt{glmnet} as well.
\begin{enumerate}
\item Initialize $\beta = \beta_0$.
\item Iterate until convergence:
\begin{enumerate}
\item Update $\eta$ by $\eta = X\beta$.
\item Update $P$ by 
\[
p_{i,m} = \frac{\operatorname{exp}\left(\eta_{i,m}\right)}{\sum_{l=1}^M \operatorname{exp}\left(\eta_{i,l}\right)}.
\]
\item Set $t = 2\max_{i,j}\left\{p_{i,j}\left(1-p_{i,j}\right)\right\}$ and set $R = \left(Y - P\right)/t$.
\item Iterate until convergence: for $k=1,\ldots,p$
\begin{enumerate}
\item Update $R_{-k}$ by
\[ 
R_{-k} = R + X_{\cdot k}\beta_{k \cdot}.
\]
\item Update $\beta_{k \cdot}$ by
\[
\beta_{k \cdot} \leftarrow
\frac{1}{\left\|X_{\cdot k}\right\|_2^2}\left(1 -
  \frac{\lambda/t}{\left\|X_{\cdot k}^{\top}R_{-k}\right\|_2}\right)_{+}X_{\cdot k}^{\top}R_{-k}.
\]
\item Update $R$ by
\[
R = R_{-k} - X_{\cdot k}\beta_{k\cdot}.
\]
\end{enumerate}
\end{enumerate}
\end{enumerate}

\subsection{Path solution}
Generally we will be interested in models for more than one value of $\lambda$. As per usual in \texttt{glmnet} we compute solutions for a path of $\lambda$ values. We begin with the smallest $\lambda$ such that $\beta = 0$, and end with $\lambda$ near $0$. By initializing our algorithm for a new $\lambda$ value at the solution for the previous value, we increase the stability of our algorithm (especially in the presence of correlated features) and efficiently solve along the path. It is straightforward to see that our first $\lambda$ value is
\[
\lambda_{\textrm{max}} = \max_k\left\|X_{\cdot k}^{\top}\left(Y - P_0\right)\right\|_2
\]
where $P_0$ is just a matrix of the sample proportions in each class. We generally do not solve all the way to the unregularized end of the path. When $\lambda$ is near $0$ the solution is very poorly statistically behaved and the algorithm  takes a long time to converge --- any reasonable model selection criterion will choose a more restricted model. To that end, we choose $\lambda_{\textrm{min}} = \epsilon\lambda_{\textrm{max}}$ (with $\epsilon = 0.05$ in our implementation) and compute solutions over a grid of $m$ values with $\lambda_j = \lambda_{\textrm{max}}(\lambda_{\textrm{min}}/\lambda_{\textrm{max}})^{j/m}$ for $j=0,\,\ldots,\,m$.

\subsection{Strong rules}
It has been demonstrated in \citet{tibshirani2012strong} that using a prescreen can significantly cut down on the computation required for fitting lasso-like problems. Using a similar argument as in \citet{tibshirani2012strong}, at a given $\lambda = \lambda_j$ we can screen out variables for which
\[
\left\|X_{\cdot k}^{\top}R\left(\lambda_{j-1}\right)\right\|_2 \leq \alpha\left(2\lambda_j - \lambda_{j-1}\right)
\]
where $R\left(\lambda_{j-1}\right) = Y - X\hat{\beta}\left(\lambda_{j-1}\right)$ for the gaussian case and $R\left(\lambda_{j-1}\right) = Y - \hat{P}\left(\lambda_{j-1}\right)$ for the multinomial case. Now, these rules are unfortunately not ``safe'' (they could possibly throw out features which should be in the fit, though in practice they essentially never do). Thus, at the end of our algorithm we must check the Karush-Kuhn Tucker optimality conditions for all the variables to certify that we have reached the optimum (and potentially add back in variables in violation). In practice, there are very rarely violations.

\section{Elastic net}
We have seen that in some cases performance of the lasso can be improved by the addition of an $\ell_2$ penalty. This is known as the elastic-net \citep{ZH2005}.   Suppose now we wanted to solve the elastic-net problem
\[
\min \frac{1}{2}\left\|Y - X\beta\right\|_F^2 +
\lambda\alpha\sum\left\|\beta_{k\cdot}\right\|_2 + \frac{\lambda(1-\alpha)}{2}\left\|\beta\right\|_F^2
\]
We can again solve one row of $\beta$ at a time. The row-wise solution
satisfies
\[
\left\|X_{\cdot k}\right\|_2^2\hat{\beta}_{k\cdot}^{\top} -
X_{\cdot k}^{\top}R_{-k} + \lambda\alpha S(\hat{\beta}_{k\cdot}) + \lambda(1-\alpha)\beta_{k\cdot}^{\top} = 0
\]
Again, simple algebra gives us that
\begin{equation}\label{eq:enet}
\hat{\beta}_{k\cdot} = \frac{1}{\left\|X_{\cdot k}\right\|_2^2 + \lambda(1-\alpha)}\left(1 - \frac{\lambda\alpha}{\left\|X_{\cdot k}^{\top}R_{-k}\right\|_2}\right)_{+}X_{\cdot k}^{\top}R_{-k}
\end{equation}
Thus, the algorithm to fit the elastic net for multiresponse regression is exactly as before with step (b)
replaced by (\ref{eq:enet}). We can similarly apply this to multinomial regression, and replace step (ii) of multinomial algorithm with our new update.

\section{Timings}
We timed our algorithm on simulated data, and compared it to the \texttt{msgl} package \citep{msgl}, which implements an alternative algorithm to solve this problem, described in \citet{vincent2012}. We also compare to a similar algorithm and implementation in our package (\texttt{glmnet}) for the usual  multinomial lasso regression without grouping. Both \texttt{glmnet} implementations are written in \texttt{R} with the heavy lifting done in \texttt{Fortran}. The \texttt{msgl} code interfaces from \texttt{R}, but all of optimization code is written in \texttt{c++}. All simulations were run on an Intel Xeon $X5680$, $3.33$ ghz processor. Simulations were run with varying numbers of observations $n$, features $p$, and classes $M$ for a path of $100$ $\lambda$-values with ($\lambda_{\textrm{min}} = 0.05\lambda_{\textrm{max}}$) averaged over $10$ trials. Features were simulated as standard Gaussian random variables with equicorrelation $\rho$. In all simulations, we set the true $\beta$ to have iid $N(0, 4/M^2)$ entries in its first $3$ rows, and $0$ for all other entries.\\

\begin{table}
\begin{center}
\begin{tabular}{l|rr}
  \hline
  & $\rho$ = 0 & $\rho = 0.2$ \\
  \hline\\
  &\multicolumn{2}{c}{$n=50,\, p=100,\, M=5$}\\
\\
  \hline
  \texttt{glmnet} grouped  & 0.32 & 0.36 \\
  \texttt{msgl} grouped & 2.86 & 2.32 \\
  \texttt{glmnet} ungrouped & 0.13 & 0.12 \\
  \hline\\
  &\multicolumn{2}{c}{$n=100,\, p=1000,\, M=5$}\\
\\
  \hline
  \texttt{glmnet} grouped  & 1.03 & 1.21\\
  \texttt{msgl} grouped & 12.14 & 8.95 \\
  \texttt{glmnet} ungrouped & 0.56 & 0.45\\
  \hline\\
  &\multicolumn{2}{c}{$n=100,\, p=5000,\, M=10$}\\
\\
  \hline
  \texttt{glmnet} grouped  & 3.71 & 4.58\\
  \texttt{msgl} grouped & 66.57 & 44.64\\
  \texttt{glmnet} ungrouped & 2.90 & 2.49\\
  \hline\\
  &\multicolumn{2}{c}{$n=200,\, p=10000,\, M=10$}\\
\\
  \hline
  \texttt{glmnet} grouped  & 12.14 & 16.59 \\
  \texttt{msgl} grouped & 284.43 & 186.05 \\
  \texttt{glmnet} ungrouped & 9.32 & 7.70\\
   \hline
\end{tabular}
\vspace{4mm}
\caption{Timings in seconds for \texttt{msgl} and \texttt{glmnet} implementations of grouped and \texttt{glmnet} implementation of ungrouped multinomial lasso, for a path of $100$ $\lambda$-values, averaged over $10$ trials, for a variety of $n,\,p,\,M$, and $\rho$.}
\label{tab:timings}
\end{center}
\end{table}
From Table~\ref{tab:timings} we can see that for the grouped problem, \texttt{glmnet} is an order of magnitude faster than \texttt{msgl}. Also, though compared to the ungrouped multinomial lasso we do take a little bit of a hit, our grouped algorithm can still solve very large problems quickly, solving gene-expression sized problems in under a minute.

\section{Discussion}
We have given an efficient group descent algorithm for fitting the group-penalized multiresponse and multinomial lasso models and empirically shown the efficiency of our algorithm. It has also been included in the current version (1.8-2) of the \texttt{R} package \texttt{glmnet}.

\bibliographystyle{abbrvnat}
\bibliography{man}

\end{document}